\documentclass[prodmode,acmtosn]{acmsmall} 

\usepackage[ruled]{algorithm2e}
\usepackage{amssymb,subfigure,epstopdf,array,color,citesort,url,fancyhdr,algorithmic,amsmath}

\SetAlFnt{\small}
\SetAlCapFnt{\small}
\SetAlCapNameFnt{\small}
\SetAlCapHSkip{0pt}
\IncMargin{-\parindent}

\newcommand{\notsure}[1]{\textcolor{black}{#1}}
\newcommand{\ds}[1]{\textcolor{black}{#1}}

\newcommand{\ve}[1]{\boldsymbol{#1}}

\begin{document}


\title{Socially Optimal Coexistence of Wireless Body Area Networks Enabled by a Non-Cooperative Game}
\author{JIE DONG, DAVID B. SMITH and LEIF W. HANLEN
\affil{National ICT Australia (NICTA $^\dag$), Australian National University}}


\begin{abstract}
In this paper, we enable the coexistence of multiple wireless body area networks (BANs) using a finite repeated non-cooperative game for transmit power control. With no coordination amongst these personal sensor networks, the proposed game maximizes each network's packet delivery ratio (PDR) at low transmit power. In this context we provide a novel utility function, which gives reduced benefit to players with higher transmission power, and a subsequent reduction in radio interference to other coexisting BANs. Considering the purpose of inter-BAN interference mitigation, PDR is expressed as a compressed exponential function of inverse signal-to-interference-and-noise ratio (SINR), so it is essentially a function of transmit powers of all coexisting BANs. It is shown that a unique Nash Equilibrium (NE) exists, and hence there is a subgame-perfect equilibrium, considering best-response at each stage independent of history. In addition, the NE is proven to be the socially optimal solution across all action profiles. Realistic and extensive on- and inter-body channel models are employed. Results confirm the effectiveness of the proposed scheme in better interference management, greater reliability and reduced transmit power, when compared with other schemes that can be applied in BANs.
\end{abstract}




\acmformat{Jie Dong, David B. Smith, Leif W. Hanlen, 2014. Socially Optimal Coexistence of Wireless Body Area Networks Enabled by a Non-Cooperative Game.}

\begin{bottomstuff}
$^\dag$NICTA is funded by the Australian Government as represented by the Department of Broadband, Communications and the Digital Economy and the Australian Research Council through the ICT Centre of Excellence program.

Authors' address: NICTA, Tower A, 7 London Circuit, Canberra, ACT 2601, Australia.

Author's email addresses: \{Jie.Dong; David.Smith; Leif.Hanlen\}@nicta.com.au;
\end{bottomstuff}

\maketitle

\section{Introduction}

Over the last decade, the urgent concerns for public health and physical well-being has led to a dramatic development of various types of wearable technologies. Among them, wireless body area networks (BANs), also known as wireless body sensor networks, are very promising because of their affordability, flexibility and convenience. A typical BAN comprises several on-body or implanted sensors that monitor physiological parameters and a gateway device, which is connected to the internet \cite{tg6_d}. In the medical domain, this architecture ensures the user's information is kept up-to-date at \ds{their} health service\ds{s} centre, which benefits comprehensive healthcare. BANs are also used in areas such as consumer fitness, entertainment, gaming and the military.

The increasing popularity of BAN has resulted in a rapid growth in active devices in the last five years \cite{Lewis2008,ABIResearch}. Considering the typical circumstance of using BANs, it is often necessary to have several BAN systems operating in close proximity to each other. However, the transmission power of BAN sensor nodes is strictly limited so as to prolong their \ds{operation life-}time \cite{tg6_d}, which however makes the system vulnerable to radio interference from other BANs. Therefore, BANs coexistence and the resultant inter-BAN interference is a major issue, which can cause severe performance degradation and packet loss as shown in \cite{6346463,6133624,6216832}. In addition, inter-BAN interference is a cause of energy wastage of sensor nodes trying to compete for better signal-to-interference-plus-noise ratios (SINRs).

There have been many studies on effective interference management schemes in wireless networks \cite{douros2011review}. Many techniques involves transmission power control that is based on a centralized \cite{Zander1992_centralised,Wu1999_centralised} and partially distributed \cite{Zander1992_PartialSIR_balancing} approach. They are effective in the context of cellular networks and general large-scale ad-hoc networks \cite{lin2006atpc} as these networks have fewer resource constraints and a stabler network topology. However, these methods are difficult to apply to BANs due to their high mobility. In \cite{Lee&Lin1996_SIRBalancing} and \cite{SINR_balancing}, fully distributed power control schemes were proposed, which we refer to as SINR-balancing in this paper. SINR balancing works well for distributed cellular mobile systems, but their performance is unknown for BANs. More recently, game-theoretic based resource allocation schemes have been widely proposed for various types of wireless networks, from spectrum allocation \cite{Candogan2010,Seneviratne2011} to transmit power control algorithms \cite{Ozel2009}.

To better model the scenario of multiple BANs coexistence as shown in Fig.\ref{fig: Multiple WBAN} and considering the difficulty in finding a global coordinator among BANs, they are modeled as rational players competing for resources (common channel) in a non-cooperative game. We employ a local non-cooperative game-based power control scheme at each BAN. It uses a novel price-dependent utility function to determine \ds{the} next superframe's transmission power\ds{, which provides a unique} Nash Equilibrium. As the IEEE 802.15.6 standard clearly specifies a maximum packet error rate of 10\%, packet delivery ratio (PDR) is used as the target utility. Hence, by maximizing the utility function, a higher PDR can be achieved. The instantaneous PDR is expressed as a compressed exponential function of SINR, which is essentially a function of all transmit powers of coexisting BANs. Obviously, by raising its transmit power, a self-interested BAN can achieve a better utility outcome if the other BANs keep their transmit powers unchanged. However, if every BANs in the range do so, both local and aggregate utilities get worse due to larger interference experienced. Therefore, a proper pricing function is employed to restrain this behavior by penalizing utilities of BANs with high transmit power. In \cite{Kazemi}, Kazemi also presented an interference mitigation scheme using game-based power control games for BANs with limited cooperation amongst networks over static channels. In our work, the dynamic body movement and shadowing  are taken into consideration, which is typically experienced on BANs\cite{Zhang2009_channelStability,Smith2009,Smith:AT:2011,Cotton2006}, and no coordination exists amongst networks.

\begin{figure}
\centerline{\includegraphics[width=0.7\columnwidth]{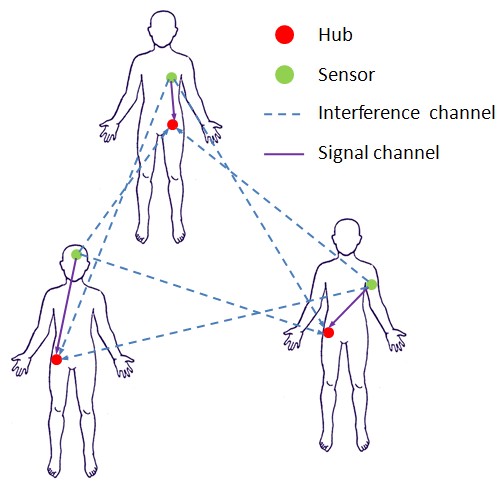}}
\caption{Code before preprocessing.}
\label{fig: Multiple WBAN}
\end{figure}

This paper makes the following novel contributions:
\begin{itemize}
\item We introduce a non-cooperative game theoretic transmit power control scheme for BANs to maximize the local packet delivery ratio across multiple BANs while reducing average transmit power.
\item We model the packet delivery ratio of a typical BAN system in terms of instantaneous SINR using a compressed exponential function.
\item The performance of the proposed power control scheme is evaluated using realistic on- and inter-body time-selective channel models, which is the typical operating environment for the BANs.
\item \ds{The implementation of this game results in a unique Nash equilibrium that additionally provides a socially optimal outcome across all coexisting BANs.}
\end{itemize}
 The rest of this paper is organized as follows. In Section II, the details of the system setup and channel models are explained. Section III defines the non-cooperative game for multiple BANs coexistence, describes our novel utility function \ds{and shows the unique Nash equilibrium, as well as the socially optimal outcome of the game}. In Section IV, the performance of this power control method is assessed and compare with other schemes. Finally, in Section V some concluding remarks are made.

\section{System model}
\subsection{Table of notations}
Table \ref{table: parameters} lists all symbols used in this paper.

\begin{table}
\tbl{Notations used in this paper\label{table: parameters}}{
  \begin{tabular}{|c|l|}
    \hline
    Notation & Meaning of the notation \\ \hline
    $T_d$ & Time length of superframe \\ \hline
    $T_{idle,i}$ & Time length of the idle period between two consecutive superframes for BAN i \\ \hline
    $T_{beacon,i}$ & Time between two consecutive beacons for BAN i \\ \hline
    $N_c$ & Number of orthogonal channels used in the inter-BAN TDMA scheme \\ \hline
    M & Total number of BANs coexisting in close proximity (including both active and idle BANs)\\ \hline
    m & Number of BANs transmitting concurrently \\ \hline
    $S,\ S_{active},\ S_{idle}$ & Set of all coexisting BANs; Set of all BANs transmitting concurrently; Set of all idle BANs\\ \hline
    $p_i,\ p_{-i}$ & Transmission power of a sensor in BAN i; Transmission power of all other active BANs except BAN i\\ \hline
    $P_i^{min},\ P_i^{max}$ & Minimum and maximum transmission power of sensors in BAN i\\ \hline
    $h_i^i$ & On-body channel gain between sensor and hub in BAN i\\ \hline
    $h_j^i$ & Inter-body channel gain between sensor in BAN j and hub in BAN i\\ \hline
    $\sigma$ & Average additive white Gaussian noise power \\ \hline
    $\gamma_i$ & SINR over a packet received at the hub of BAN i \\ \hline
    pdr & Packet delivery rate \\ \hline
    $a_c,\ b_c,\ a,\ b$ & Parameters of the SINR vs. PDR approximation and simplified approximation \\ \hline
    $dist_o$,\ $dist_{ij}$ & Reference distance; distance between subjects (BANs) i and j \\ \hline
    $(x_i,\ y_i)$ & Coordinates of BAN i  \\ \hline
    $U(\cdot)$ & Utility function of the proposed cooperative power control game \\ \hline
    $\sum\limits_{i=1}^m U_i(\cdot)$ & Social welfare\\ \hline
    $w_i,\ v_i,\ d_i$ & Parameters of the utility function\\ \hline
\end{tabular}}
\end{table}

\subsection{BANs coexistence based on probability of overlapping}
We consider multiple subjects in the proximity of each other, each wearing a typical star-topology BAN. Time division multiple access (TDMA) is employed as the intra-BAN access scheme. Here, the period of time, between hub broadcasting a beacon and all sensors completing transmission in a round-robin fashion, is defined to be a superframe $T_d$. It is followed by an idle period $T_{dile}$, during which sensors are inactive and waiting for next beacon. In terms of the inter-BAN access scheme, a unsynchronized TDMA scheme is used as we assume no coordination exists amongst networks \cite{Zhang2010}. In this concept, the channel is temporally divided into $N_c$ orthogonal channels (time slots) and each time slot has a length of $T_d$. BAN $i$ chooses its transmission starting time of a superframe independently and randomly following a uniform distribution over [0, $(N_c-1)T_d$].

\begin{align}
    &T_{beacon,i} = T_d + T_{idle,i},\\
    &Pr(T_{idle,i} = t) = \begin{cases}
         \ \frac{1}{(N_c-1)T_d}, &t\in[0, (N_c-1)T_d] \\
         \ 0, &t\ \mathrm{otherwise} \end{cases},
\end{align}
where $T_{beacon,i}$ is the time between two consecutive beacons from the hub of BAN $i$. Here, we assume $S_M$ is the set of all M BANs located in close proximity, which consists of a subset of all active BANs $S_{active}$ and a subset of all idle BANs $S_{idle}$.
\begin{align}
    &S = S_{active} \cup S_{idle}, \\
    &and\ |S|= |S_{active}| + |S_{idle}| = M .\nonumber
\end{align}
where $|\cdot|$ represents the number of element in a set. Therefore, the probability of $m$ BANs transmitting concurrently $Pr(|S_{active}| = m)$ is calculated as follow:
\begin{align}
	&Pr(|S_{active}| = m) = \binom{M}{m}\left(\frac{2}{N_c}\right)^m\left(1-\frac{2}{N_c}\right)^{M-m}, \label{equ:coexistence} \\
    &\textrm{where}\ m \in [1,M], N_c \ge 2\nonumber
\end{align}
The probability of a BAN actively transmitting with respect to any other BANs is $\frac{2}{N_c}$ due to the lack of synchronization of inter-BAN TDMA scheme. Fig.\ref{fig: Coexisting probability} shows the probability of $m$ BANs transmitting concurrently with 8 BANs $(M = 8)$ coexisting. Different colors correspond to the varying number of orthogonal channels available. In this study, the proposed power control game is simulated over many occasions with different channels, wherein each occasion the value of $m$ is chosen randomly following the probability distribution \{$Pr(|S_{active}| = m)$\}. By introducing this probability distribution, it also models the mobility of BANs for multiple networks coexistence as BANs can leave and also enter the area of interest.

\begin{figure}[t!]
\centerline{\includegraphics[width=0.8\columnwidth]{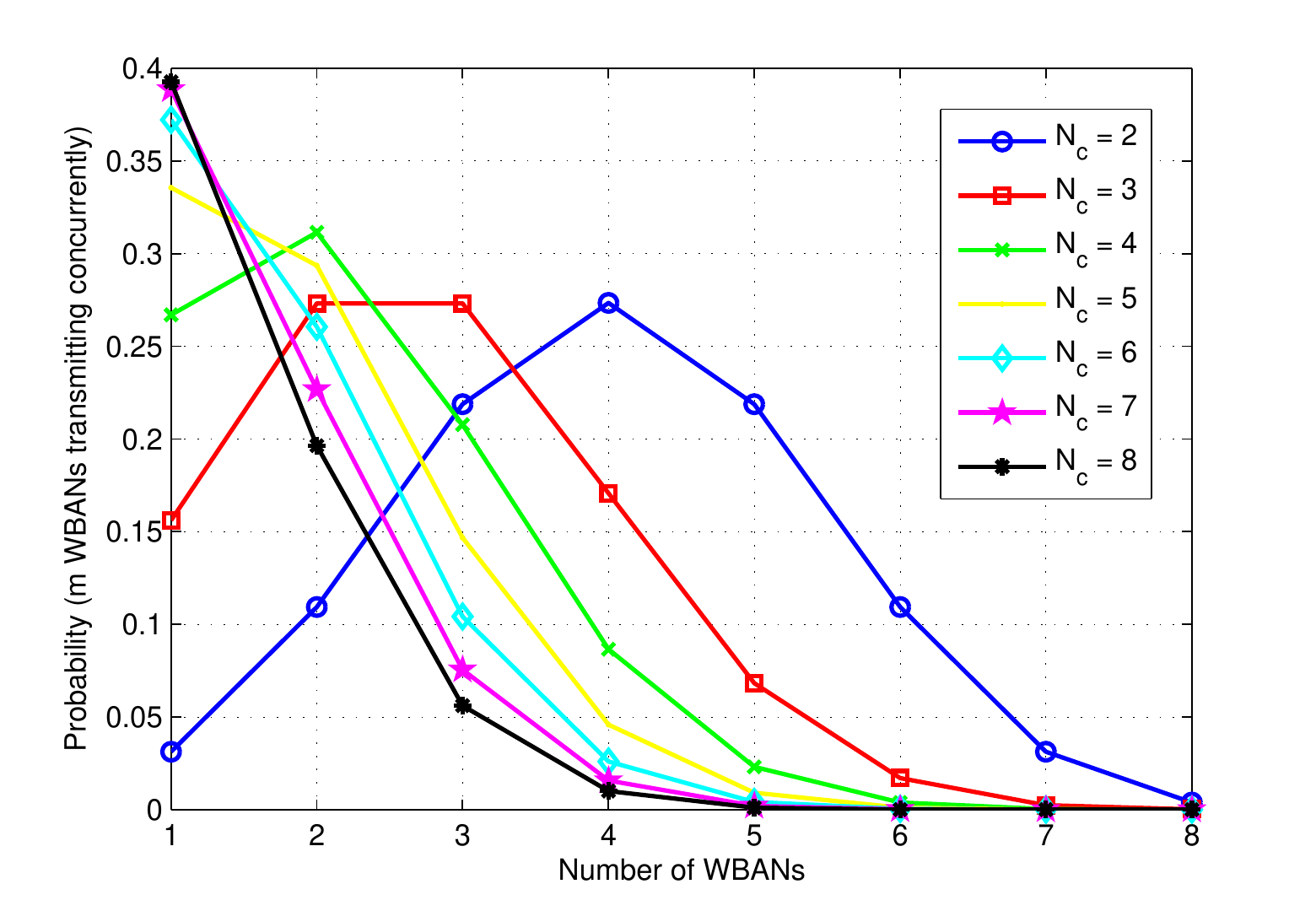}}
\caption{Multiple BANs coexistence probability}
\label{fig: Coexisting probability}
\end{figure}

\subsection{SINR-based packet delivery ratio}
At any time, a sensor in each BAN transmits concurrently with sensor nodes in other $m-1$ BANs. Hence, for each BAN, the hub receives not only the signal packet from its own sensor node, but also $m-1$ interfering signal packets. Therefore, the SINR over a signal packet at the hub of BAN $i$, $\gamma_i(\tau)$, is calculated as follows:

\begin{align}
	&\gamma_i(\tau) = \frac{p_i(\tau)|h_i^i|^2}{\sum_{j=1,j\ne i}^m p_j(\tau)|h_j^i|^2+\sigma^2},
\label{equ:SINR}
\end{align}
where $p_i(\tau)$ is the transmission power of a sensor in the $i$th BAN at time $\tau$; $h_j^i$ represents the average channel gain across a packet time from the sensor in BAN $j$ to the hub in BAN $i$, in other words, the interference channel from interferer $j$ to network $i$. In terms of $h_i^i$, it is the average on-body channel gain from the sensor to its connected hub in BAN $i$ in the same time interval.

\begin{figure}[t!]
\centerline{\includegraphics[width=0.8\columnwidth]{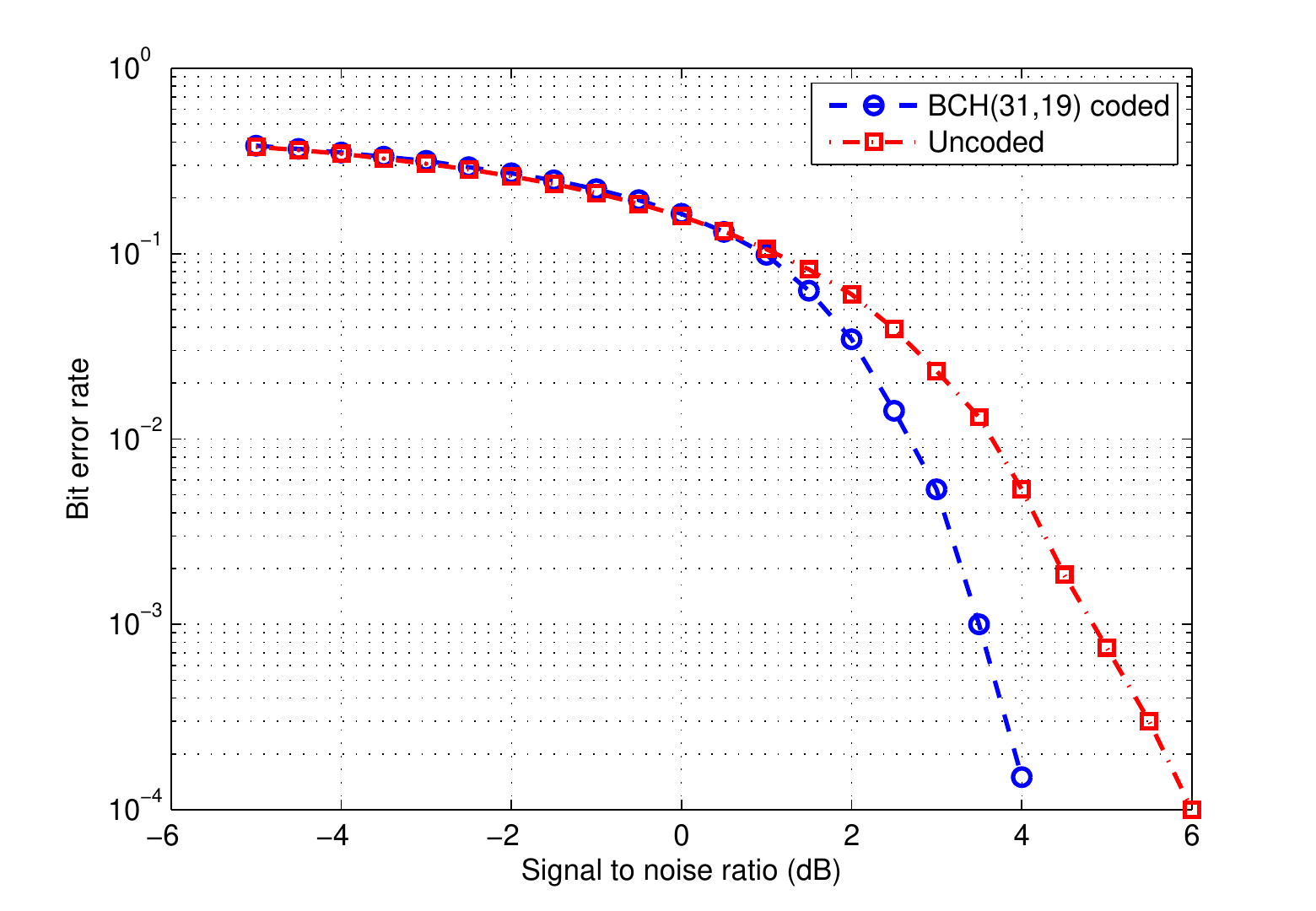}}
\caption{BCH coding gain}
\label{fig: BCH Coding Gain}
\end{figure}

Observing that the graph of general PDR vs. SINR is a sigmoidal function, it is possible to express the PDR as a compressed exponential function of inverse SINR, $1/\gamma$ \cite{SmithTVT}. In (\ref{equ:PDR}), $\gamma$ is calculated as (\ref{equ:SINR}) and $a_c$ and $b_c$ are constant parameters depends on particular modulation, coding scheme and packet length. Complying with the IEEE 802.15.6 standard \cite{tg6_d}, BCH(31,19) coding and DPSK/BPSK modulation scheme are applied with a packet length of 256 bytes. It is found that BCH(31,19) provides about 2 dB channel coding gain as shown in Fig.\ref{fig: BCH Coding Gain}, this advantage is considered when estimating the PDR vs. SINR relation.\footnote{BCH(31,19) is used as an example here. There are other non-IEEE 802.15.6 compliant coding schemes which can provide higher coding gain.} With a root-mean-square error of the approximation less than 0.006, Fig.\ref{fig: SINR fit} shows the comparison between approximated and simulated PDR vs. SINR relation for DPSK and BPSK respectively. For later simplicity of analysis, we can rearrange the equation to be expressed as (\ref{equ:simplePDR}), where $a = -(1/a_c)^{b_c}$ and $b = -b_c$. The values of $a$ and $b$ are given as in Table \ref{parameters} for both DPSK and BPSK.

\begin{align}
	pdr &= \exp\left(-\left(\frac{1}{\gamma a_c}\right)\right)^{b_c} \label{equ:PDR}\\
     &= \exp\left(a \gamma^b\right), \label{equ:simplePDR}
\end{align}

\begin{figure}[t!]
\centerline{\includegraphics[width=0.8\columnwidth]{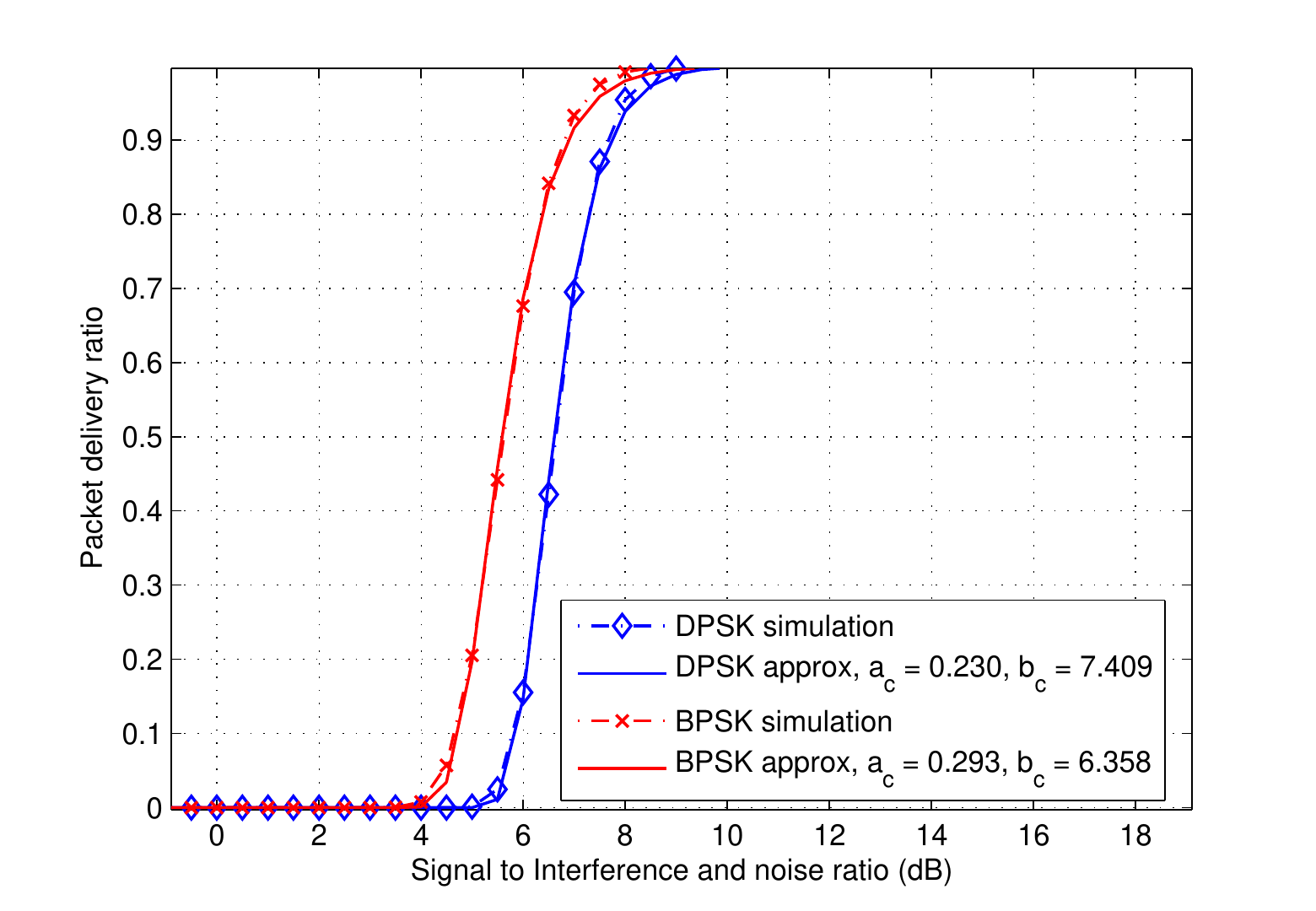}}
\caption{Simulated SINR vs. Approximated SINR}
\label{fig: SINR fit}
\end{figure}

\begin{table}
\tbl{Estimated parameters in the compressed exponential function (\ref{equ:PDR}) and its simplified form (\ref{equ:simplePDR})\label{parameters}}{
  \begin{tabular}{|c|c|c|c|c|}
    \hline
    Modulation & $a_c$ & $b_c$ & $a$ & $b$ \\ \hline
    DPSK & 0.230 & 7.409 & $-$337.2164 & $-$7.4540 \\ \hline
    BPSK & 0.293 & 6.358 & $-$30.0512 & $-$6.3470 \\
  \hline
\end{tabular}}
\end{table}

\section{Non-Cooperative Power control Game}
\label{Sec:game definition}
\subsection{Game Definition}
In this multiple BANs coexistence game, each BAN network makes an independent decision on the transmit power of the next packet based on its current SINR. Here each BAN is treated as a player in a non-cooperative repeated game $G = \{N, \mathbf{P}, U\}$, where:
\begin{enumerate}
    \item $N = \{1, 2, ... , m\}$ is a finite set of players whom indexed by $i$. $N$ represents the same set as $S_{active}$;
    \item $\mathbf{P}$ represents the global strategy space, which is the Cartesian product of all players' strategy spaces, i.e. $\mathbf{P} = \mathbf{P}_1 \times \mathbf{P}_2 \times ... \times \mathbf{P}_m$. The pure strategy set of player $i$, $\mathbf{P}_i$, is a finite set of discrete transmit powers in the range of $[P_i^{min}, P_i^{max}]$.  The action of player $i$ at any time(stage) $\tau$ is denoted as $p_i(\tau) \in \mathbf{P}_i$, and $\ve{p_{-i}}$ stands for the choice of transmission power of other players except player $i$;
    \item The utility function $U_i$ is defined in terms of the current transmission power $p_i(\tau)$ and packet delivery ratio PDR. Its objective is to maximize the PDR while minimizing the transmit power. It is defined as follow:
        \begin{align}
            U(p_i, pdr_i) = -p_i^{w_i}-\frac{d_i}{pdr_i^{v_i}},
            \label{equ:utility function_1}
        \end{align}
        where $pdr_i$ is a function of SINR and thus a function of the transmit powers of all players according to (\ref{equ:PDR}) and (\ref{equ:SINR}). Hence, $U(p_i,pdr_i)$ can be rearranged and expressed as $U(p_i,p_{-i})$. The exponents $v_i > 0$ and $w_i > 0$ depend on the particular network configuration, and can be varied accordingly. The weighting factor $d_i > 0$ can be adjusted depending on the current network status. At the end of every time slot, players (BANs) update their transmit power levels to maximize the outcome from applying the utility function based on the latest transmit power and the current SINR:
        \begin{align}
            &p_i(\tau+1) = \arg \max (U(\mathbf{P_i}, p_{-i})), \label{equ:maximize}\\
            \textrm{where}\ &\mathbf{P}_i = \{p_i | p_i \in [P_i^{min}, P_i^{max}]\}, \forall i \in N \nonumber
        \end{align}

\end{enumerate}

\subsection{Nash Equilibrium}
An important condition for the non-cooperative game to converge is that a unique Nash Equilibrium (NE) exists. The existence and uniqueness of the Nash Equilibrium of the defined game are proved as follow.
    \begin{definition}
        \label{def: NE}
        The action profile $\mathbf{p}^* = (p_1^*, p_2^*, ..., p_m^*) \in \mathbf{P}$ is a Nash Equilibrium if, for all players, $p_i^*$ is a best response to $\ve{p_{-i}}^*$. In the other words, there exists $U_i(p_i^*,\ve{p_{-i}}^*) \ge U_i(p_i,\ve{p_{-i}}^*)$ for any choice of $p_i \in \mathbf{P_i}$.
    \end{definition}
\smallskip
    \begin{theorem}
        At least one Nash Equilibrium exists for the non-cooperative finite repeated game $G = \{N, \mathbf{P}, U\}$ proposed here. 
    \end{theorem}

\begin{proof}
\begin{enumerate}
\item $[P_i^{min}, P_i^{max}]$ is a nonempty, convex and compact subspace of a Euclidean space $\mathbb{R}^m$.
\item The utility function (\ref{equ:utility function_1}) is continuous in the domain $[P_i^{min}, P_i^{max}]$. This can be shown by taking the first derivative of the utility function and substituting (\ref{equ:simplePDR}) and (\ref{equ:SINR}):

    \begin{align}
        \frac{\delta U_i}{\delta p_i}
        &= -w_ip_i^{w_i-1} + \frac{d_iv_i}{pdr_i^{v_i+1}}\frac{\delta pdr_i}{\delta p_i}, \label{equ:NE_1}\\
        &= -w_ip_i^{w_i-1} + \frac{d_iv_i}{pdr_i^{v_i+1}}pdr_iab\gamma_i^{b-1}\frac{\delta \gamma_i}{\delta p_i}, \label{equ: NE_2}\\
        &= -w_ip_i^{w_i-1} + ab\gamma_i^{b-1}\frac{d_iv_i}{pdr_i^{v_i}}\frac{|h_i^i(k_i)|^2}{I_{-i}}, \label{equ: NE_3}
    \end{align}
    where $I_{-i}$ is the interference and noise power experienced at the hub of player i. Based on (\ref{equ:SINR}), we can derive the relation $\frac{\gamma_i}{p_i} = \frac{|h_i^i(k_i)|^2}{I_{-i}}$. Therefore,
    \begin{align}
        \frac{\delta U_i}{\delta p_i} = -w_ip_i^{w_i-1} + ab\gamma_i^{b-1}\frac{d_iv_i}{pdr_i^{v_i}}\frac{\gamma_i}{p_i}
    \end{align}
    Since $p_i \in [P_i^{min}, P_i^{max}]$ is real and $pdr_i$ is non-zero according to the approximation shown in (\ref{equ:simplePDR}), the first derivative function is defined. Therefore, Theorem 1 is proved.
\end{enumerate}
\end{proof}

\begin{theorem}\label{NEtheor}
       The Nash Equilibrium at each stage in the non-cooperative power control game $G$  is unique, when $d_i > 0$. With the unique Nash equilibrium at each stage, which is independent of history, there is a unique sub-game perfect equilibrium.
\end{theorem}

\begin{proof}
To show the Nash Equilibrium point $p_i$ is unique in the range of $[P_i^{min}, P_i^{max}]$, it is sufficient to check the concavity of the utility function $U(p_i, pdr_i)$ by taking the second derivative.
        \begin{align}
            \frac{\delta^2 U_i}{\delta p_i^2}=
            -w_i(w_i-1)p_i^{w_i-2} + c_i\frac{\gamma_i^{b-2}}{pdr_i^{v}}\{(b-1)-abv_i\gamma_i^b\}
            \label{equ: second de};
        \end{align}
where $c_i = abd_iv_i\frac{|h_i^i|^4}{I_{-i}^2}$. Due to the fact that $p_i$ is always positive, the first part of (\ref{equ: second de}) has a negative value as long as the constraint of exponent $w_i > = 1$ is satisfied. In addition, since $v_i >0$, and $a$ and $b$ are both negative and SINR,$c_i$ and PDR are always positive, the second part of (\ref{equ: second de}) is always negative. The addition of these two parts means the second derivative $\frac{\delta^2 U_i}{\delta p_i^2} < 0$ in the range of $[P_i^{min}, P_i^{max}]$. Therefore, the function $U(p_i,pdr_i)$ is concave and has a local maximum at $p_i^*$ which occurs at the point $\frac{\delta U}{\delta p_i} = 0$. In other words, the Nash Equilibrium at each stage of this game is unique. Furthermore at any given stage, it can be seen that this equilibrium is independent of the history, hence there is a sub-game perfect equilibrium.
\end{proof}

\subsection{\ds{Social Optimality and} Pareto Optimality}
As the existence and uniqueness of the Nash Equilibrium (NE) of the defined non-cooperative game \ds{has been proved}, the efficiency of \ds{operating at} the NE \ds{now needs to be determined}. It should be noted that a NE, as per definition \ref{def: NE}, is the best response in a single player's point of view, given the decisions of the other players. However, in the proposed non-cooperative game, \ds{there is imperfect} information \ds{as the} other BANs' transmit power for the current time slot is unknown. Therefore, it predicts other BANs' action\ds{s} based on the \ds{latest} aggregate interference power received. Additionally, as players in the game act in their own self interest, there is no guarantee that the decision is optimal from a social point of view, or \ds{even from an individual BAN's perspective}. Therefore, the efficiency of the NE solution is characterized as follows:
\begin{definition}
\label{def: Social optimality}
A joint action profile $\mathbf{p}^* = \left(p_1^*,p_2^*, ..., p_m^*\right) \in \mathbf{P}$ is a socially optimal (efficient) outcome if \ds{there is} no other joint action profile $\mathbf{p}$ such that $\sum\limits_{i=1}^m U_i(p_i,\mathbf{p_{-i}}) \geq \sum\limits_{i=1}^m U_i(p_i^*,\mathbf{p_{-i}^*})$.
\end{definition}
In Definition \ref{def: Social optimality}, $\sum\limits_{i=1}^m U_i(p_i,\mathbf{p_{-i}})$ is the social welfare of the game. As $\mathbf{p}^*$ leads to the maximal social welfare, it is social\ds{ly} optimum and therefore \ds{is also} Pareto efficient \cite{SocialOptimal}. In this paper, we would like to show that the action profile $\mathbf{p}^* = \left(p_1^*,p_2^*, ..., p_m^*\right)$ at the NE point is indeed \ds{socially optimal}. This can be observed by comparing the social welfare, i.e. aggregate utility outcome of (\ref{equ:utility function_1}),  obtained with all different action profiles $\mathbf{p} \in \{\mathbf{P} = \mathbf{P}_1 \times \mathbf{P}_2 \times ... \times \mathbf{P}_m\}$.

\begin{figure}[t!]
\centerline{\includegraphics[width=0.8\columnwidth]{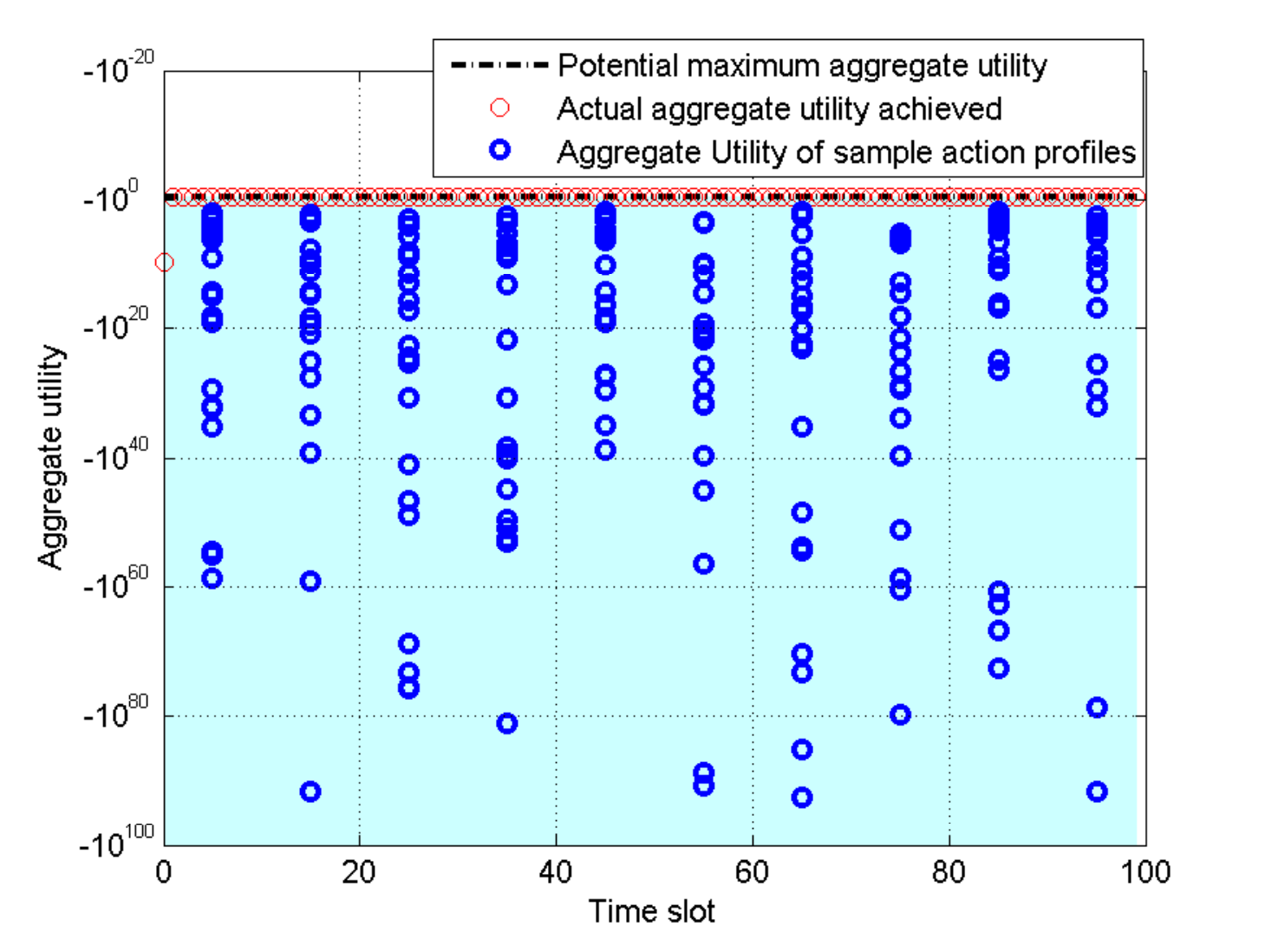}}
\caption{Social welfare outcome of the game}
\label{fig: optimal}
\end{figure}
\ds{
\begin{proposition}
\ds{When all BANs} are operating at the unique Nash equilibrium shown in Theorem \ref{NEtheor}, then according to Definition \ref{def: Social optimality} the outcome of the proposed power control game is socially optimal.
\end{proposition}}
\begin{proof}Following (\ref{equ:utility function_1}), \ds{it may appear that the} maximum social welfare that the game can \ds{achieve} is when all BANs transmit at the lowest power, i.e. $-$30dBm, and meanwhile achieve a PDR of 1. However, due to the on-body channel attenuation and inter-BAN interference, \ds{it is generally impossible to achieve this}. With exhausted comparing among all possible action profiles $\mathbf{p}$ at each time slot, the maximum possible social welfare is plotted as the black dashed line in Fig.\ref{fig: optimal}. It matches the social welfare reached with actual action profile \ds{$\mathbf{p}^*$} determined at the \ds{unique} NE, which is indicated as red dot\ds{s} in Fig.\ref{fig: optimal}. Additionally, Fig.\ref{fig: optimal} also shows the social welfare of \ds{a large agglomerate of randomly sampled action profiles} in the blue shaded area, which are all below the actual social welfare of the game. It is noted that the aggregate utility achieved at time slot $0$ is not the optimal social welfare, since the initial transmit power of all BANs are randomly chosen in the range of $[-30, 0]$ dBm. According to all \ds{possible} observation\ds{s}, the Nash Equilibrium is therefore social\ds{ly} optimal\ds{, and hence also Pareto optimal}.
\end{proof}

\subsection{Algorithm Description}
In this section, an iterative and distributed power control game that determines each BAN's transmission power at time $\tau$ is described as in Alg.\ref{Alg: power control}.
\begin{algorithm}
\caption{The proposed distributed non-cooperative power control game}

\begin{enumerate}
    \item In BAN $i$, a sensor transmits a packet with $p_i(\tau)$. When $\tau = 0$, $p_i(0) \in [P_i^{min}, P_i^{max}]$ is randomly chosen. If $\tau > 0$, $p_i(\tau-1)$ is determined in the previous iteration at $\tau-1$;\\
    \item Hub in BAN $i$\ ($i \in N$) calculates the SINR $\gamma_i(\tau)$ of the received packet as (\ref{equ:SINR});\\
    \item Estimate the instantaneous PDR $pdr_i(\tau)$ with $\gamma_i(\tau)$ based on the PDR vs. SINR approximation (\ref{equ:simplePDR});\\
    \item As (\ref{equ:maximize}),determine the transmission power at $\tau+1$, $p_i(\tau+1) \in [P_i^{min},P_i^{max}]$, which gives the maximum value of $U(p_i, pdr_i(\tau))$;
\end{enumerate}
\label{Alg: power control}
\end{algorithm}

\section{Channel model}
\label{sec: channel model}
To evaluate the performance of the proposed game-based power control algorithm, extensive on- and inter-body channels are modeled. We simulate the scenarios in which a random number of BANs are coexisting and moving in arbitrary directions. Since the same network topology is used in each BAN, it is a reasonable assumption that on-body channels are \notsure{independent and identically distributed for all players}. The gamma distribution can characterize the general everyday on-body channel of a BAN \cite{Smith:AT:2011}, so gamma fading with a mean 60 dB attenuation, with shape parameter of 1.31, and scale parameter of 0.562, which considers the effect of body shadowing and BAN channel dynamics, is employed.

In terms of the inter-body interfering channel, we start with representing the movement of a player by a series of $(x(\tau), y(\tau))$ coordinates updated every 1 ms. The initial positions of players are randomly chosen within a $6 \times 6 \textrm{m}^2$ square area which corresponds to the requirements in the standard \cite{tg6_d}. During their movement, a random small turning angle is introduced to model a realistic walking pattern of an individual. In addition, an average walking speed of 3m/s with 0.2m/s standard deviation is applied. This walking model enables us to calculate the distance between two players $i$ and $j$ at any instant throughout the simulation. The channel attenuation is then calculated based on path loss model, body shadowing and also small scale fading,
\begin{align}
    h_i^j &= A_t(dist_0/dist_{ij})^{(2.7/2)}A_{BS}A_{SC},
    \label{equ: channel attenuation}
\end{align}
assuming a path loss exponent of 2.7 between BANs\footnote{\ds{Although other common path loss exponents $\geq$ 2 for the environments in which BANs are co-located are equally applicable}.}. \notsure{$dist_{ij}$ represents the distance between players i and j, and }the reference distance $dist_0 = 5m$ corresponds to a channel attenuation $A_t$ of 54~dB. We consider the average case where body shadowing contributes approximately $A_{BS} = 45$ dB attenuation and adopt a Jakes' model with Doppler spread of 1.1 Hz as the Rayleigh distribution for the small scale fading $A_{SC}$ between BANs.

\section{Simulation analysis}
One individual repeated game consists of 100 repeated game-playing stages. 20 sets of channels are generated based on the description in Section \ref{sec: channel model}. Because of the random moving velocity, the walking pattern varies between different channel model sets. With each set, the same 100-stage game is played on 50 occasions, using different segments of the data. Therefore, a total of 1000 games, each with 100 stages, are conducted. Here, a maximum of 8 BAN networks locating in the vicinity is simulated, i.e. $M$ = 8 in (\ref{equ:coexistence}). At any time during the simulation, a various number of BANs are active with the others idle. The actual number of networks transmitting concurrently follows the probability distribution ${P_m}$ described in (\ref{equ:coexistence}). The case of $m = 0$ is neglected as the case of no network transmitting is irrelevant. In terms of inter-BAN TDMA scheme, we assume 4 orthogonal channels ($N_c = 4$) are used. In addition, to investigate the effect of a given number of coexisting BANs on the performance of the proposed algorithm, constant numbers of BANs coexisting are also simulated, with the number coexisting from 2 to 8 each run on 1000 occasions.  The exponents in the utility function $v$ and $w$ are set to be 4 and 1, which give the best outcomes for the game.

Based on the same configuration and channel models, we compare the proposed game with some other schemes commonly applied in BAN. The comparison is based on two criteria -- (i) percentage of BANs reaching the target PDR; and (ii) transmission power at each stage. The result is averaged across all 1000 games. Complying to the IEEE 802.15.6 standard \cite{tg6_d}, the target PDR is set to be 0.9. The rest of the schemes are \textbf{Sample-and-Hold}\cite{Smith:ICC:2011,Dong:ICC:2014}, \textbf{SINR-Balancing}\cite{SINR_balancing}\cite{Lee&Lin1996_SIRBalancing} and constant transmission power at \textbf{0/$-5$/$-10$ dBm}. Here in \textbf{Sample-and-Hold} current SINR is used for each BAN to set its' next transmission power, unlike in \cite{Dong:ICC:2014}, where it is done with respect to channel gain. Here in \textbf{Sample-and-Hold} the transmission power is adjusted based on the latest packet's SINR to attempt to achieve the target SINR for the next packet.

\begin{figure}[t!]
\centering
\subfigure[BPSK PDR]
{\includegraphics[width=0.8\columnwidth]{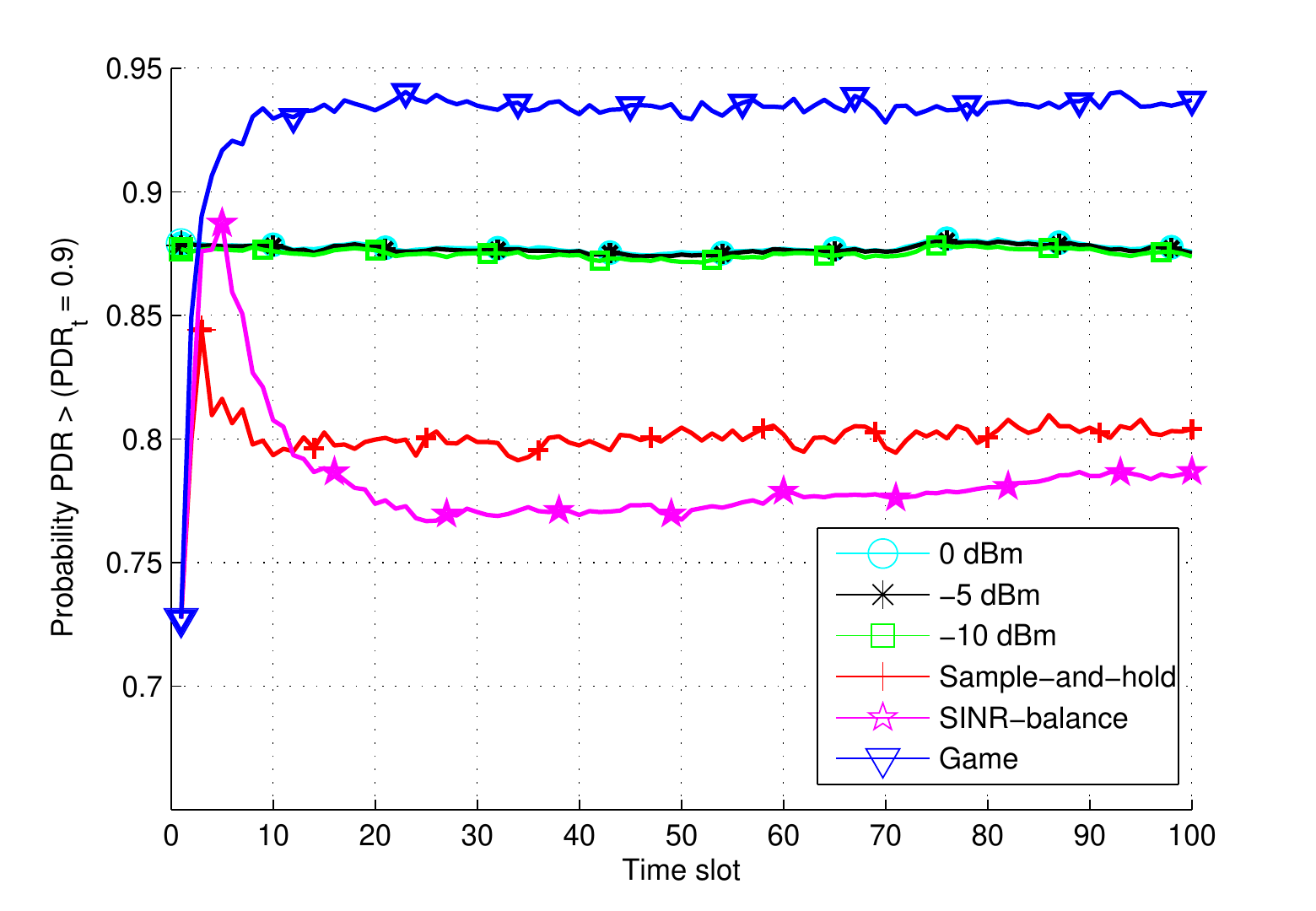}
\label{fig:subfig1a}}
\subfigure[BPSK power]
{\includegraphics[width=0.8\columnwidth]{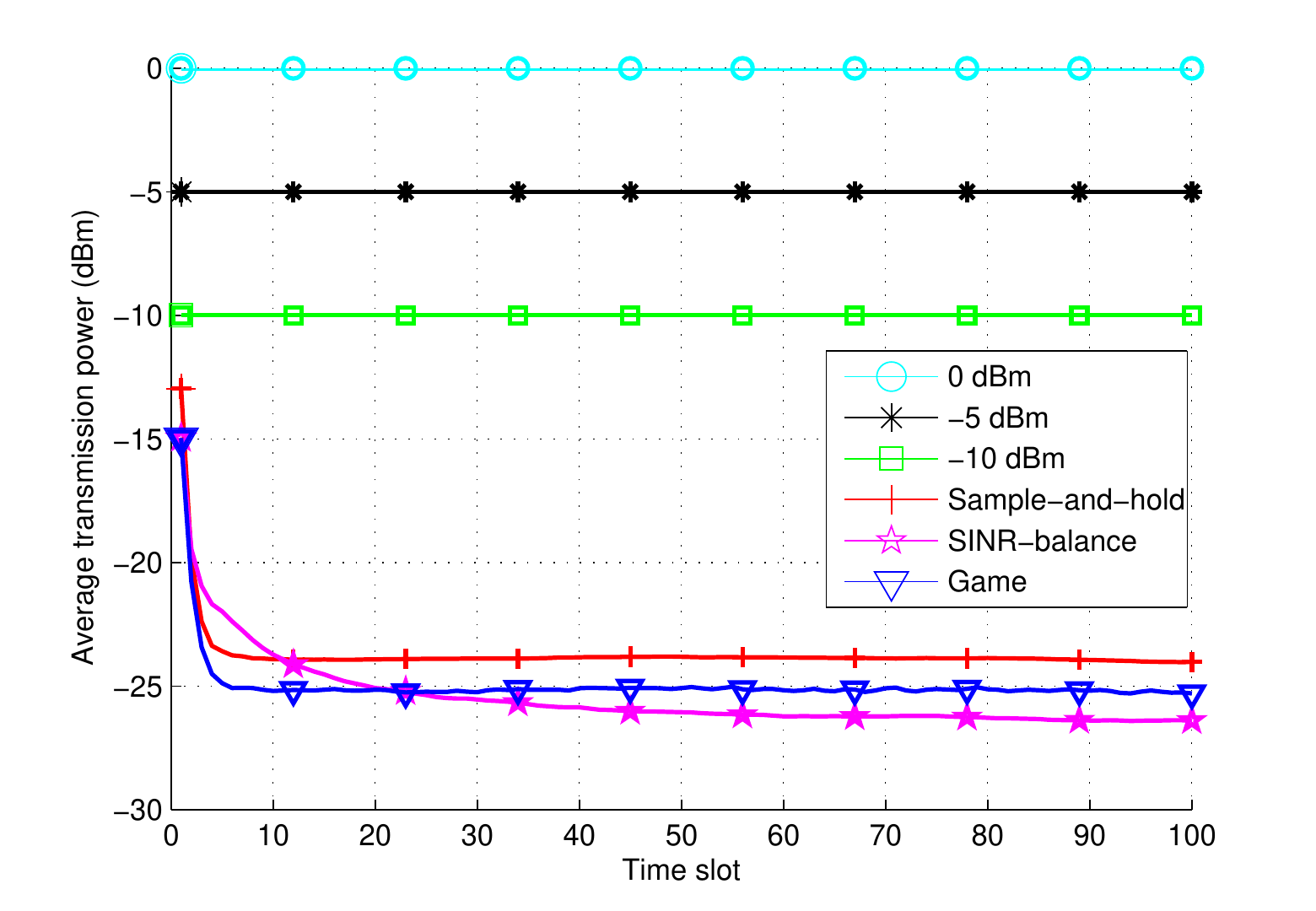}
\label{fig:subfig1b}}
\caption{BPSK Simulation results}
\label{fig:BPSK}
\end{figure}

\begin{figure}[t!]
\centering
\subfigure[DPSK PDR]
{\includegraphics[width=0.8\columnwidth]{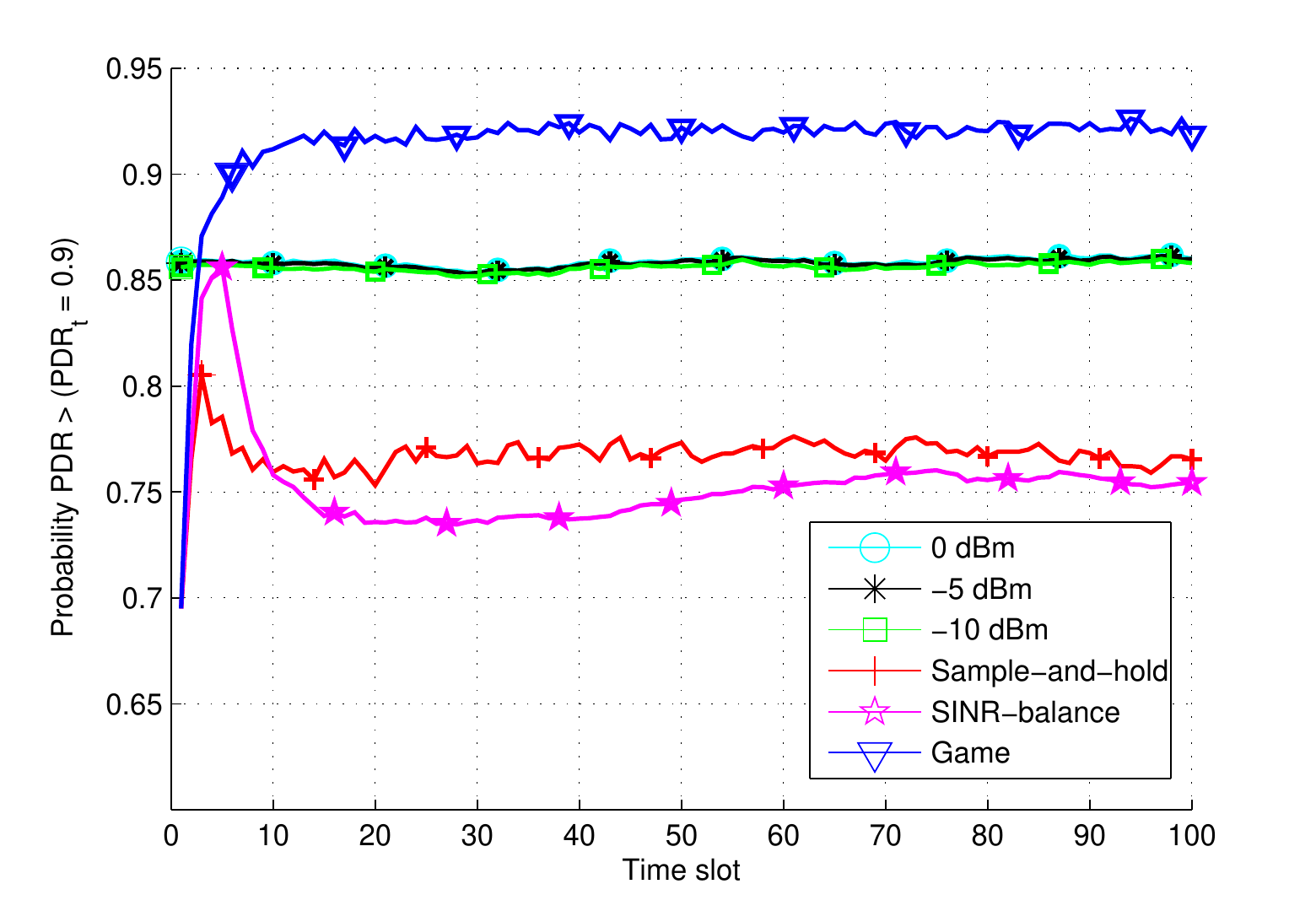}
\label{fig:subfig2a}}
\subfigure[DPSK power]
{\includegraphics[width=0.8\columnwidth]{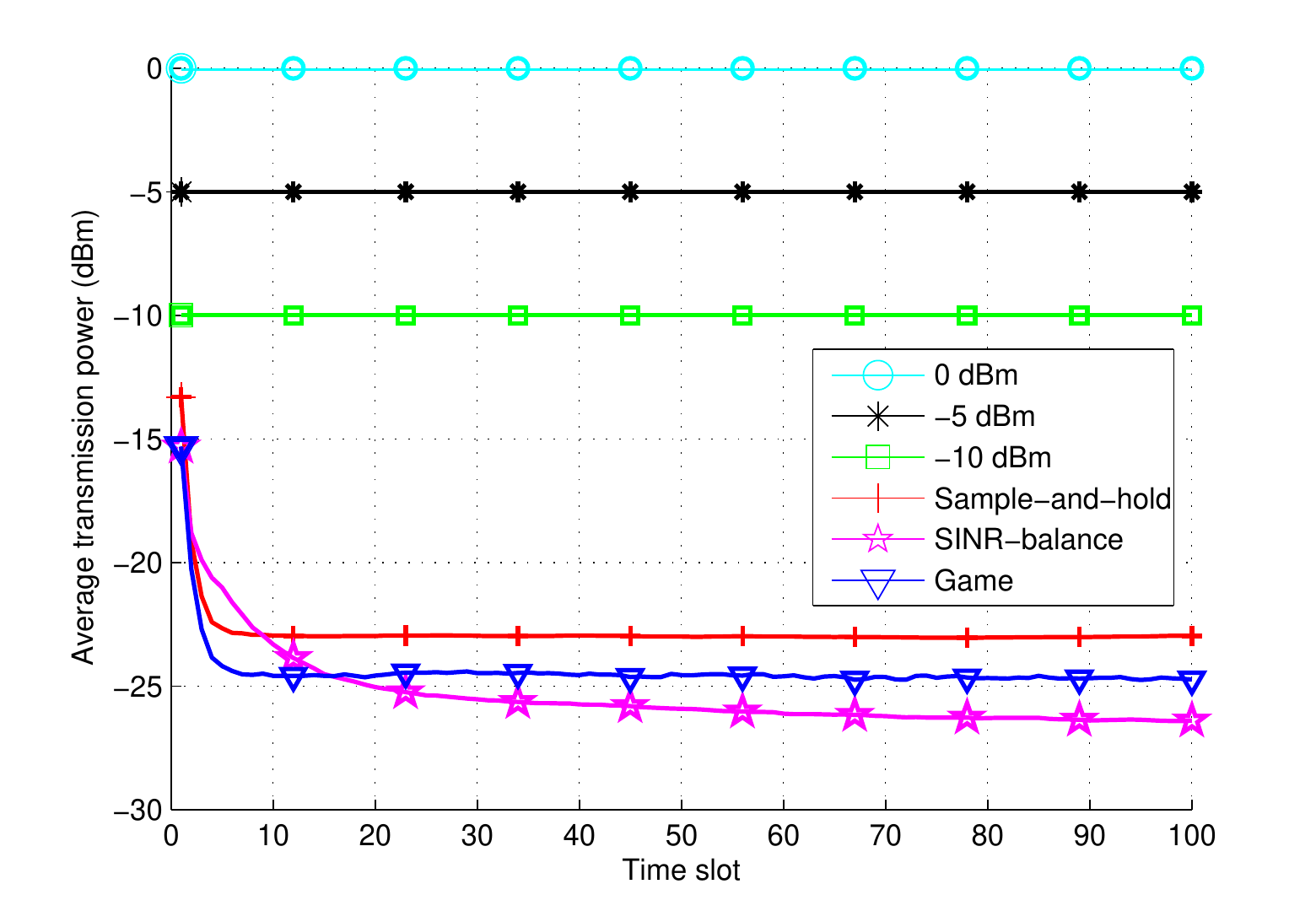}
\label{fig:subfig2b}}
\caption{DPSK Simulation results}
\label{fig:DPSK}
\end{figure}

All power control methods are applied and plotted in Fig.\ref{fig:BPSK} and \ref{fig:DPSK}, using BPSK and DPSK modulations respectively for the case where number of concurrently transmitting BANs follows the probability distribution ${P_m}$ (\ref{equ:coexistence}). From Fig.\ref{fig:subfig1a}, it is shown that using the proposed game-based power control method, approximately 93\% of the BANs are able to achieve the target PDR of 0.9 while this number is only 80\% and 77\% for Sample-and-Hold and SINR-balancing methods. Constant transmission at different power level shows similar performance to each other with about 87\% achieving target PDR. In terms of the time taken to converge to the steady-state minimized transmit power, the proposed method and Sample-and-hold achieve this 16 time slots ahead of SINR-balancing. The short convergence time of the first two algorithms ensures that they can quickly respond to time-variations in the target channel and also interference, which is typical for BAN operation. In terms of the output transmit power shown in Fig.\ref{fig:subfig1b}, the game  has an average of $-$25 dBm while sample-and-hold is about $-$23 dBm. SINR-balancing has 2 dB less in average transmission power compared with our proposed method. However, with its poor performance in percentage of BANs reaching the target and slow response time, it is not a suitable choice for enabling BANs coexistence. Further the game has at least 15 dB less average transmit power than the constant power transmission. Similar output transmit powers can be observed in Fig.\ref{fig:subfig2a} when DPSK is employed. In this case, the percentage of BANs reaching the target PDR is 92\%, 85\%, 76\% and 74\% for the proposed algorithm, constant power transmission, Sample-and-Hold and SINR-balancing respectively.

\begin{figure}[t!]
\centering
\subfigure[Percentage of BANs reaching target PDR]
{\includegraphics[width=0.8\columnwidth]{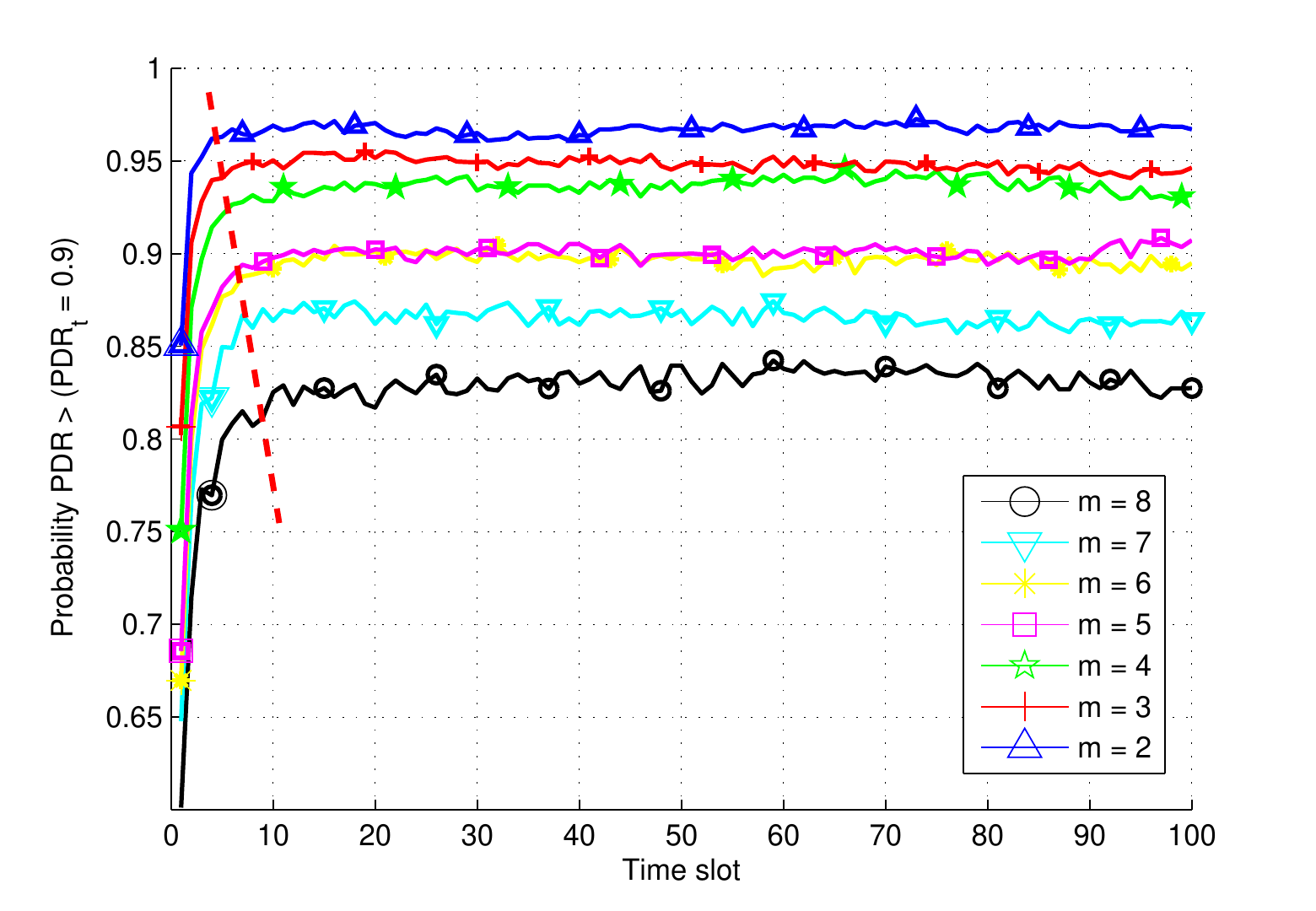}
\label{fig:overall}}
\subfigure[Average transmission power of BANs]
{\includegraphics[width=0.8\columnwidth]{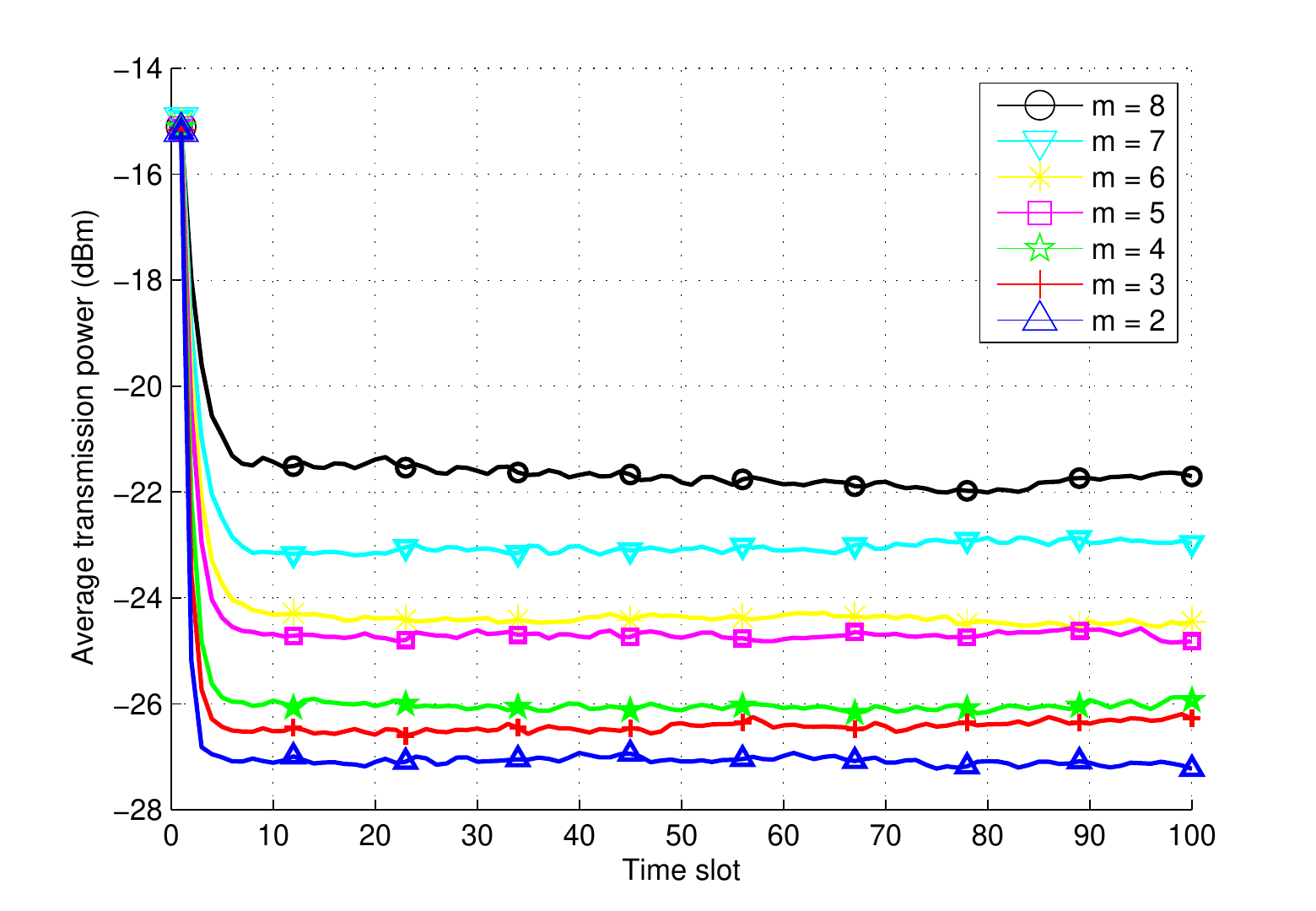}
\label{fig:overall power}}
\caption{Different number of BANs coexisting scenarios}
\label{fig:differentWBANs}
\end{figure}

Next we show the effect of changing the number of players on the performance of the proposed power control game. Fig.\ref{fig:differentWBANs} shows the comparison when the number $m$ of coexisting BANs is fixed, $m \in [2,8]$, with respect to the previous criteria. The same simulation parameters (exponents $v=4$ and $w=1$, weighting factor calculation $d_i$) are used for different values of $m$. It is observed that the average percentage of BANs reaching the target PDR of 0.9 decreases with increasing $m$, from 97\% to 83\%. In Fig.\ref{fig:overall}, the intercept point of the red broken line and each solid line indicates the approximate time slot for game convergence, which shows that the more players that join the game the longer it takes for the game to converge. In Fig.\ref{fig:overall power}, we can see that transmit power rises from $-$27 dBm to $-$21 dBm as $m$ increases from 2 to 8. Note that the performance of the proposed algorithm for different values of $m$ always outperforms any of the previously described schemes described previously in the case of the average percentage of BANs reaching the target. Although SINR balancing uses slightly smaller transmit power it sacrifices a lot in terms of reliability.
\section{Conclusion}
Wireless Body Area Networks (BANs) have been pervasively used in many areas. For these personal sensor networks, with no global coordination amongst multiple closely-located networks, there can be severe performance degradation. For better interference management, a non-cooperative power control game has been proposed, to enable coexistence amongst BANs. In this game, a novel utility function , which constrains output transmit power is applied for each player. The unique Nash Equilibrium, which is also a social\ds{ly} optimal solution, and sub-game perfect equilibrium leads to a converg\ds{ed} outcome after a \ds{small} number of stages, in terms of all BANs reaching target packet delivery ratio at the lowest possible transmit power. Based on extensive simulation over different instantiations of a realistic channel model, our proposed scheme can achieve a significantly higher number of BANs \ds{more rapidly} reaching target PDR than other power control methods that are typically employed in distributed wireless networks. In addition, the lower circuit power consumption as result of lower transmit power using the proposed game, can significantly prolong the lifetime of \ds{the battery of the} sensor radio. Finally, increasing the number of coexisting BANs \ds{only} degrades the performance of the proposed power control game by a small amount, \ds{still outperforming} other \ds{feasible} methods.

\bibliographystyle{ACM-Reference-Format-Journals}
\bibliography{refList}

%
%
%
%
%
%
%

\end{document}